\documentclass[11pt,envcountsame,envcountsect]{llncs}
\usepackage[letterpaper,hmargin=1.05in,vmargin=1.05in]{geometry}
\usepackage{graphicx}
\usepackage{amssymb}

\newcommand{\ka}{\ensuremath{(a,b,\varepsilon)}--key-agreement\ }

\title{Breaking One-Round Key-Agreement Protocols in the Random Oracle Model}
\author{Miroslava Sot\'{a}kov\'{a}}
\institute{Department of Computer Science\\
\email{mirka@cs.au.dk}}

\begin{document}
\maketitle

\begin{abstract}
In this paper we study one-round key-agreement protocols analogous to Merkle's puzzles in the random oracle model. The players Alice and Bob are allowed to query a random permutation oracle $n$ times and upon their queries and communication, they both output the same key with high probability. We prove that Eve can always break such a protocol by querying the oracle $O(n^2)$ times. The long-time unproven optimality of the quadratic bound in the fully general, multi-round scenario has been shown recently by Barak and Mahmoody-Ghidary. The results in this paper have been found independently of their work. 
\end{abstract}

\section{Introduction}

In this paper, we find a tight upper bound on the number of queries needed to break a key-agreement protocol in the random oracle model. Key-agreement protocols called Merkle's puzzles were constructed by Ralph Merkle in 1974 but only published in 1978~\cite{Mer78}. They are one of the earliest examples of public-key cryptographic protocols. 

The key-agreement ala Merkle between Alice and Bob proceeds as follows: Alice constructs a large number of puzzles, each of them being possible to solve with Bob's computational resources. In other words, all of them are in the form of an encrypted message with an unknown key that is short enough to allow for a brute force attack. After receiving the message from Alice, Bob chooses one puzzle uniformly at random and solves it. The solution contains an identifier and a key. Bob encrypts the identifier with the key, and announces it back to Alice. The solution of the puzzle solved by Bob becomes Alice's and Bob's secret key. Since the puzzle's identifier is sent to Alice as a message encrypted with a key that is unknown to Eve, the eavesdropper's best strategy to attack such a protocol is to solve as many puzzles as possible. To achieve constant probability of success, Eve has to solve a constant fraction of them, which might require much more computational power than what is needed by the legitimate players.

In a similar way we construct a key-agreement protocol in the random oracle scenario, where the computational difficulty of a key-agreement protocol is expressed by the number of oracle queries that Alice and Bob make in order to agree on a secret key. Instead of creating many puzzles, Alice queries the oracle in many positions that are unknown to both Bob and Eve, and sends the images of the queried elements to Bob. Bob queries the oracle in sufficiently many positions to get a collision with Alice's set of queries with high probability. He recognizes the collision from Alice's message and reports it back to Alice by its identifier -- the oracle image. The pre-image becomes Alice's and Bob's secret key. In addition to only few queries, the communication gives Eve a little information about the key, since the oracle is random. With the same number of queries as Bob, she would find a collision with Alice's set of queries with high probability, but not necessarily the one found by Bob. Hence, finding the right element might require significantly more oracle queries than Alice and Bob needed to agree on the key.
 
Until recently, the best upper bound on Eve's number of queries needed to break such protocols have been shown by Impagliazzo and Rudich~\cite{IR91}. They prove that in any key-agreement protocol based on a random-permutation oracle, where Alice and Bob agree on the secret key in $n$ rounds in such a way that they query only one query per round (normal form of a protocol), Eve needs at most $O(n^3)$ oracle queries to output a secret key guess that matches with Bob's secret key with the same probability as Alice's key does. For a protocol in the general form, $O(n^6)$ are sufficient for an attack, which can be proven by showing that any protocol can be transformed into its normal form with at most quadratic blow-up in the number of oracle queries made by the players. In~\cite{IR91}, the question is studied in the larger context to show that possibility of secure key-agreement relative to some random permutation oracle implies ${\rm P}\neq {\rm NP}$. In other words,  proofs for showing that existence of one-way functions implies existence of secure key-agreement do not relativize.

The bound from~\cite{IR91} has been improved recently by Barak and Mahmoody-Ghidary~\cite{barak-2008} who show that in fact, $O(n^2)$ are sufficient for Eve's attack. 

In our paper we deal with one-round key-agreement protocols where Alice and Bob query the oracle $a$ and $b$ times, respectively. Such protocols form a subset of protocols whose normal form consists of $a+b$ rounds. We prove the tight -- $O((a+b)^2)$ upper bound on the number of queries Eve needs to break the protocol. 

Throughout the paper, we use the following notation: $\chi_A$ denotes the characteristic function of set $A$, $\mathbb{E}(X)$ denotes the mean value of random variable $X$, and $\mathbb{E}^c(X)$ denotes the 
mean value of $X$, conditioned on information $c$. 

\section{One-Round Key-Agreement Protocols}
In this section, we model one round key-agreement protocols
between Alice and Bob. We assume that Alice, Bob, and an eavesdropper Eve
have access to an oracle computing a random permutation $f$ on $\{1,\dots,n\}$. We define a one-round key-agreement protocol as follows:
\medskip

{\bf Protocol 1}

Given $n\in  \mathbb{N}$ and an oracle computing  a random permutation $f$ on $\{1,\dots,n\}$,
\begin{enumerate}
\item Alice queries the oracle $f$ in positions $\mathcal{A}_1\in \{1,\dots,n\}^{\leq a}$, computes a message $c_A$ and sends it to Bob.
\item Bob, given $c_A$, queries the oracle $f$ in positions $\mathcal{B}\in \{1,\dots,n\}^{\leq b}$, computes message $c_B$ and sends it to Alice. Bob generates the secret key $k_B\in\{0,1\}^{\ell}$, $k_B=g_B(\mathcal{B},f(\mathcal{B}),c=(c_A,c_B),R_B)$, where $R_B$ denotes his local randomness.
\item Alice, given $c$, queries the oracle in positions $\mathcal{A}_2\subseteq \{1,\dots,n\}$ such that for $\mathcal{A}:=\mathcal{A}_1\cup \mathcal{A}_2$, $|\mathcal{A}|\leq a$, and generates the secret key $k_A\in\{0,1\}^{\ell}$, $k_A=g_A(\mathcal{A},f(\mathcal{A}),c,R_A)$, where $R_A$ denotes her local randomness.
\end{enumerate}

We denote by \ka any one-round key-agreement protocol
defined as above and satisfying the following condition: $\Pr{[k_A\neq k_B]}\leq \varepsilon$ where $\varepsilon<1$ is a constant.
\medskip

Notice that $a$ and $b$ are functions of $n$, but for simplicity we refer to them by $a$ and $b$, instead of using $a(n)$ and $b(n)$, if the latter one is not explicitly needed.  
\medskip
Since key-agreement protocols take place between players Alice and Bob
sharing no initial secret, the key generation mechanism must involve
only common queries to the oracle $f$. We say that Eve breaks the protocol if she outputs a string that agrees with Bob's key with the same probability as Alice does.

\begin{lemma}
\label{intersection}
In order to break an \ka protocol it is sufficient for Eve to query all intersection queries of Alice and Bob used for the generation of Alice's secret key. 
\end{lemma}

\begin{proof}
Eve querying all elements in $\mathcal{A}_1\cap \mathcal{B}$ can construct a permutation $f'$ matching with $f$ on $\mathcal{E}$ (Eve's queries), and a set $\mathcal{A}'_1$ of queries to the oracle computing $f'$ such that $c_A=c_{A'_1}$ and $f'$ is consistent with $c_B$. Therefore, after querying $\mathcal{B}$, Bob has exactly the same view about $\mathcal{A}_1$ as he has about $\mathcal{A}'_1$. Eve constructs the set $\mathcal{A}'_2$ according to $\mathcal{A}'_1$ and $c$, and then ``queries" the $f'$-oracle on the positions in $\mathcal{A}'_2$. Finally, she generates her secret key $k_E=g_A(\mathcal{A}',f(\mathcal{A}'),c,R_{A'})$, where $\mathcal{A}':=\mathcal{A}'_1\cup \mathcal{A}'_2$. From Bob's point of view, both $k_E$ and $k_A$ are generated from the same set $\mathcal{K}\subseteq \mathcal{A}\cap \mathcal{B}$, i.e. $\Pr[k_B=k_E]=\Pr[k_B=k_A]$.
\qed

\end{proof}

\section{Proof of the Quadratic Upper Bound}

We will consider the following attack of an \ka protocol:
\begin{enumerate}
\item Eve repeats Bob's querying strategy $\gamma a$ times for some constant $\gamma$ (i.e. makes $\gamma ab$ oracle queries) in order to query all queries in $\mathcal{A}_1\cap \mathcal{B}$ with constant probability
\item Eve extracts the position of the $\mathcal{A}_2$-queries from $c_B$ and queries the oracle on these positions ($a$ oracle queries)
\end{enumerate}

Next we prove that with the proposed strategy Eve breaks the protocol with constant probability. 

\begin{lemma}
\label{A1}
By repeating Bob's strategy independently $5a$ times, Eve finds all elements in $\mathcal{A}_1\cap \mathcal{B}$ with constant probability.
\end{lemma}

\begin{proof} Let $A$ and $B$ denote the random variables associated with Alice querying the elements in $\mathcal{A}_1$ and Bob querying the elements in $\mathcal{B}$, respectively. Let $E$ denote the random variable associated with the set of Eve's queries $\mathcal{E}$. W.l.o.g., assume that for $x,y\in\{1,\dots,n\}, x\leq y$, $P_{\chi_{B}(x)|c_A}(1)\leq P_{\chi_{B}(y)|c_A}(1).$

Define $\mathcal{A}_1^0:=\mathcal{A}_1$, $\mathcal{B}^0:=\mathcal{B}$, $A^0=A$, $B^0:=B$, $s_0:=\mathbb{E}^{c_A}(|A_1\cap B|)$, and $n_0:=n$. In the $i$-th step, define $n_{i+1}$, $\mathcal{A}_1^{i+1}$, $\mathcal{B}^{i+1}$, $A^{i+1}$, $B^{i+1}$, $s_{i+1}$ in order to satisfy the following:
$$\forall x\in\{n_{i+1}+1,\dots,n_i\}:\ P_{\chi_{B}(x)|c_A}(1)\geq \frac{s_i}{2a},$$
$\mathcal{A}_1^{i+1}:=\mathcal{A}_1\setminus \{n_{i+1}+1,\dots,n\}$, $\mathcal{B}^{i+1}:=\mathcal{B}\setminus \{n_{i+1}+1,\dots,n\}$, let $A^{i+1}$, $B^{i+1}$ denote the corresponding random variables, and set $s_{i+1}:=\mathbb{E}^{c_A}(|A_1^{i+1}\cap B^{i+1}|).$ 

Furthermore, consider $u$ such that
$$\Pr[\mathcal{A}_1\cap \mathcal{B}\subseteq \{n_u+1,\dots,n\}|c_A]\geq \frac{1}{2}.$$
\medskip

First, we prove that
\begin{enumerate}
\item there exists $u\in \mathbb{N}$ with the desired property
\item $n_{i+1}<n_i$ for $i\in\{0,\dots,u-1\}$
\item $s_{i}-s_{i+1}\geq 1$ for $i\in\{0,\dots,u-1\}$ 
\item $s_u\geq 1$
\end{enumerate}
\medskip

We can write:
\begin{equation}
\nonumber
s_i=\mathbb{E}^{c_A}(|A_1^{i}\cap B^i|)=\sum_{\mathcal{A},|\mathcal{A}|\leq a}P_{A^i|c_A}(\mathcal{A})\sum_{\mathcal{B},|\mathcal{B}|\leq b}P_{\mathcal{A}\cap B^i|c_A}(|\mathcal{A}\cap \mathcal{B}|)
\end{equation}
hence, there exists at least one $\mathcal{A}\subseteq \{1,\dots,n\}^{\leq a}$ such that $\sum_{\mathcal{B},|\mathcal{B}|\leq b}P_{\mathcal{A}\cap B^i|c_A}(|\mathcal{A}\cap \mathcal{B}|\geq s_i$. 

Let us choose one such $\bar{\mathcal{A}}$. Then
$$s_i\leq\sum_{\mathcal{B}, |\mathcal{B}|\leq b}P_{\bar{\mathcal{A}}\cap B^i|c_A}(B)(|\bar{\mathcal{A}}\cap \mathcal{B}|
= \sum_{\mathcal{B}, |\mathcal{B}|\leq b}\sum_{x\in |\bar{\mathcal{A}}\cap \mathcal{B}|} P_{B^i|c_A}(\mathcal{B})
= \sum_{x\in \bar{\mathcal{A}}}\sum_{\mathcal{B}:\ x\in \bar{\mathcal{A}}\cap \mathcal{B}} P_{B^i|c_A}(\mathcal{B})
= \sum_{x\in\bar{\mathcal{A}}}P_{\chi_{B^i}(x)|c_A}(1).
$$

Since $|\bar{\mathcal{A}}|\leq a$, there is an $x\in \{1,\dots,n_i\}$ such that $P_{\chi_{B^i}(x)|c_A}(1)\geq \frac{s_i}{a}.$ If we remove $x\in\{1,\dots,n_i\}$ such that $P_{\chi_{B^i}(x)|c_A}(1)\geq \frac{s_i}{2a}$, 
then $s_{i+1}\leq \frac{s_i}{2}$. 

Since in every step we remove at least one $x\in\{1,\dots,n\}$, the procedure terminates after finitely many steps and therefore, $u$ is well-defined and is at most $n$. Clearly, for $s_i<1$ we have $\Pr[\mathcal{A}_1^{i+1}\cap \mathcal{B}^{i+1}=\emptyset|c_A]>\frac{1}{2}$, implying that with probability at least $1/2$ we have $\mathcal{A}_1\cap \mathcal{B}\subseteq \{n_i+1,\dots,n\}$. Therefore $s_u\geq 1$ and for $i\in\{0,\dots,u-1\}:$ 
$$s_i-s_{i+1}\geq \frac{s_i}{2}\geq \frac{s_{u-1}}{2}\geq 1.$$
\medskip

We finish the proof of the statement by showing that by repeating Bob's strategy $5a$ times independently, Eve queries all elements in $\mathcal{A}_1\cap \mathcal{B}$ with probability at least $1/8$. 

For $x\in\{n_{i+1}+1,\dots,n_i\}$, Eve does not query $x$ with probability  
$$P_{\chi_{E}(x)|c_A}(0)\leq \left(1-\frac{s_i}{2a}\right)^a\leq  e^{-s_i/2}.$$ That means that in the case where 
$$|\mathcal{A}_1^i\cap \mathcal{B}^i\cap \{n_{i+1}+1,\dots,n_i\}|\leq \frac{e^{s_i/2}}{2s_i^2},$$
the probability that Eve does not query at least one element in $\{n_{i+1},\dots,n_i\}\cap \mathcal{A}_1^i\cap \mathcal{B}^i$ is
$$\Pr\left[\prod_{x\in \{n_{i+1},\dots,n_i\}\cap \mathcal{A}_1^i\cap \mathcal{B}^i}\chi_{E}(x)=0|c_A\right]\leq \frac{e^{s_i/2}}{2s_i^2}\cdot e^{-s_i/2}=\frac{1}{2s_i^2}.$$
Since the expected number of elements in $\mathcal{A}_1^i\cap \mathcal{B}^i\cap \{n_{i+1}+1,\dots,n_i\}$ is $s_i$, Markov's inequality implies that this happens with probability at most $\frac{2s_i^3}{e^{s_i/2}}$. Hence, there exists $i$, $0\leq i<u$, such that $|\mathcal{A}_1^i\cap \mathcal{B}^i\cap\{n_{i+1}+1,\dots,n_i\}|> \frac{e^{s_i/2}}{2s_i^2}$ with probability at most $\sum_{i=0}^{u} \frac{2s_i^3}{e^{s_i/2}}$. The function $\frac{2x^3}{e^{x/2}}$ is decreasing for $x\geq 6$, yielding 
$$\sum_{i=0}^{u'-1:s_{u'}\geq 6}\frac{2s_i^3}{e^{s_i/2}}\leq\sum_{i=0}^{u'-1:s_{u'}\geq 6}(s_i-s_{i+1})\frac{2s_i^3}{e^{s_i/2}}\leq \int_{x=s_{u'}}^{\infty}\frac{2x^3}{e^{x/2}}{\rm d}x.$$
Then for $s_{u'}\geq 28$ we obtain: 
$$\sum_{i=0}^{u'-1:s_{u'}\geq 28}\frac{2s_i^3}{e^{s_i/2}}<\frac{1}{8}.$$ Furthermore, for $s_i<28$ (there are at most 5 of them, since $s_{i+1}\leq s_i/2$ and $s_u\geq 1$), the probability that $A_1^i\cap B^i\cap\{n_{i+1}+1,\dots,n_i\}$ contains more than $40s_i$ elements is at most $1/40$, by Markov's inequality. 

The probability that there exists an $i, 0\leq i<u$ such that 
$$|A_1\cap B\cap\{n_{i+1},\dots, n_i\}|> \max\{40s_i,\frac{e^{s_i/2}}{2s_i^2}\}$$ is therefore at most $\frac{1}{8}+\frac{5}{40}=\frac{1}{4}$. If this happens, we say that $\mathcal{A}_1\cap \mathcal{B}$ has a ``bad structure" for finding all its elements by Eve. 

It is sufficient for Eve to repeat Bob's algorithm $(\log{80}+3\log{s_i})a/s_i\leq 5a$ times to get all elements in $\mathcal{A}_1\cap \mathcal{B}\cap\{n_{i+1},\dots,n_i\}$, $i\geq u'$, assuming that there are no more than $40s_i$ of them, with probability at least $1-\frac{1}{2s_i^2}$. 
 
In other words, with $5a$ independent iterations of Bob's strategy, Eve does not query at least one element of well-structured $\mathcal{A}_1\cap \mathcal{B}\cap\{n_u+1,\dots,n\}=\mathcal{A}_1\cap \mathcal{B}$ with probability
\begin{eqnarray*}
&{\rm Pr}&\left[\prod_{x\in \{n_u+1,\dots,n\}\cap A_1\cap B}\chi_{E}(x)=0|\mathcal{C}_A\right]\leq \frac{1}{2}\cdot \sum_{i=0}^u \frac{1}{s_i^2}\\
&\leq& \sum_{i=0}^{u-1} (s_i-s_{i+1})\cdot\frac{1}{2s_i^2}+\frac{1}{2s_u^2}\leq \frac{1}{2}\cdot\int_{x=s_u}^{\infty}\frac{{\rm d}x}{x^2}=\frac{1}{2s_u}\leq \frac{1}{2}.
\end{eqnarray*}

Since $\mathcal{A}_1\cap \mathcal{B}\subsetneq \{n_u+1,\dots,n\}$ with probability at most $\frac{1}{2}$, and $\mathcal{A}_1\cap \mathcal{B}$ is ill-structured with probability at most $\frac{1}{4}$, $\mathcal{A}_1\cap \mathcal{B}\subseteq \{n_u+1,\dots,n\}$ and is well-structured with probability at least $\frac{1}{4}$. In this case Eve queries all intersection elements with probability at least $\frac{1}{2}$ hence, Eve finds all intersection queries of $\mathcal{A}_1$ and $\mathcal{B}$ with probability at least $\frac{1}{8}$.
\qed

\end{proof}

\begin{theorem}
Eve can break an \ka protocol with $O((a+b)^2)$ queries with constant probability.
\end{theorem}

\begin{proof}
As we claim in the proof of Lemma \ref{intersection}, Eve querying all elements in $\mathcal{A}_1\cap \mathcal{B}$ needs at most $|\mathcal{A}_2|\leq a$ queries more to generate the key that matches with Bob's secret key with the same probability as Alice's key does. Lemma \ref{A1} shows that Eve can always query all elements in $\mathcal{A}_1\cap \mathcal{B}$ with probability $1/8$ with at most $5ab$ queries. Therefore, Eve can break the protocol with constant probability with $5ab+a\in O((a+b)^2)$ oracle queries. 
\qed

\end{proof}

\section{Optimality of the Bound}

Consider the following protocol:
\medskip

{\bf Protocol 2}
\begin{enumerate}
\item Alice chooses a set $\mathcal{A}\subseteq \{1,\dots,n\}$, $|\mathcal{A}|=a=\lceil \sqrt{n}\rceil$ uniformly at random, queries the oracle for the elements of $\mathcal{A}$, and sends $c_A=\{f(x):\ x\in A\}$ to Bob.
\item Bob chooses a set $\mathcal{B}\subseteq \{1,\dots,n\}$, $|\mathcal{B}|=b=\lceil \sqrt{n}\rceil$ uniformly at random, queries the elements of $\mathcal{B}$, chooses a collision element $k\in\{f(y):\ y\in \mathcal{B}\}\cap c_A$ at random, and sends $c_B=\{f(k)\}$ to Alice. He outputs the key $k$.
\item Alice recognizes $k$ according to $c_B$ and $\mathcal{A}$, and outputs the key $k$.
\end{enumerate}
\medskip

{\bf Attack}:  With a constant probability, Bob finds at least one  collision with Alice's set of queries due to the birthday paradox, and therefore, the given protocol is an example of $(\sqrt{n},\sqrt{n},\varepsilon)$-key-agreement protocol for some constant $\varepsilon<1$.  Given just $c$, the secret key is uniformly distributed in $\{1,\dots,n\}$ and furthermore, since the oracle is random, Eve knowing the oracle image for only $o(n)$ elements still has $(1-o(1))\log n$ entropy about $f(x)$ for $x\notin \mathcal{E}$. Hence, Eve has to query the oracle in $\Theta(n)$ positions to get the right secret key with constant probability, implying that the optimal Eve's strategy to break the protocol with constant probability must involve $O(n)=O((a+b)^2)$ oracle queries. 
\section{Conclusion}
We provided an analysis of the most commonly considered attack of these type of key-agreement protocols where the attacker iterates the players' strategies with gradually updated information in the case of one-round protocols. Originally, we were hoping to generalize the result to apply in the multi-round scenario, which has been done very recently by Barak and Mahmoody-Ghidary.  

\section{Acknowledgements}
We thank Louis Salvail who introduced me into the Merkle's puzzles problem.

%\bibliography{exam2}
\end{document}